\documentclass[11pt]{amsart}

\usepackage{amsthm,amsmath,amsfonts,amssymb,graphicx}
\usepackage{amsmath,amsfonts,amsthm}

\newcommand{\R}{{\mathord{\mathbb R}}}
\newcommand{\Z}{{\mathord{\mathbb Z}}}
\newcommand{\N}{{\mathord{\mathbb N}}}
\newcommand{\C}{{\mathord{\mathbb C}}}

\newcommand{\E}{{\mathord{\mathbb E}}}

\newcommand{\SH}{{\mathord{\mathbb{SH}}}}

\newcommand{\HH}{\mathcal{H}}

\newcommand{\ran}{{\rm Ran}}

\newcommand{\ben}{\begin{displaymath}}
\newcommand{\een}{\end{displaymath}}
\newcommand{\beqn}{\begin{equation}}
\newcommand{\eeqn}{\end{equation}}
\newcommand{\beqna}{\begin{eqnarray*}}
\newcommand{\eeqna}{\end{eqnarray*}}

\def\supp{\operatorname{supp}}

\newtheorem{lemma}{Lemma}
\newtheorem{theorem}[lemma]{Theorem}

\usepackage{latexsym}
\usepackage{amssymb}

\usepackage{times,mathptm}
\usepackage{pifont}

\begin{document}

\title{On the AC spectrum of one-dimensional random Schr\"odinger operators with
matrix-valued potentials}

\author[R.~Froese]{Richard Froese}
\address{Department of Mathematics\\
University of British Columbia\\
Vancouver, British Columbia, Canada}
\email{rfroese@math.ubc.ca}
\author[D.~Hasler]{David Hasler}
\address{Department of Mathematics\\
College of William \& Mary\\
Williamsburg, Virginia, USA}
\email{dghasler@wm.edu}
\author[W.~Spitzer]{Wolfgang Spitzer}
\address{Institut f\"ur Theoretische Physik\\
    Universit\"at Erlangen--N\"urnberg\\
    Erlangen, Germany}
    \email{wolfgang.spitzer@physik.uni-erlangen.de}

\date{October 7, 2009}

\begin{abstract}
We consider discrete one-dimensional random Schr\"odinger operators with decaying
matrix-valued, independent potentials. We show that if the $\ell^2$-norm of this
potential has finite expectation value with respect to the product measure then almost surely the Schr\"odinger operator has an interval of purely absolutely continuous (ac) spectrum. We apply this result to Schr\"odinger operators on a strip. This work provides a new proof and generalizes a result obtained by Delyon,
Simon, and Souillard~\cite{DSS85}.
\end{abstract}

\maketitle

\section{Model and Statement of Results}

In this paper we are interested in the absolutely continuous (ac) spectrum of
quasi one-dimensional random Schr\"o\-dinger operators with decaying potentials.
To this end, it is convenient to formulate the problem in terms of
matrix-valued potentials on the one-dimensional lattice, $\Z$.

Let us first introduce some standard notation that is used throughout this
paper. If $H$ is an operator on some Hilbert space $\HH$, then we denote by
$\rho(H),\sigma(H),\sigma_{\rm{ac}}(H),\sigma_{\rm{ess}}(H)$ its resolvent set,
spectrum, ac spectrum, respectively its essential spectrum. By $\|H\|$ we denote the operator norm of $H$.

For some $m\in\N$, let ${\rm Sym}(m)$ denote the set of real symmetric
$m \times m$ matrices. Let $D \in {\rm Sym}(m)$ be some fixed matrix and let
$q=(q_n)_{n\in\Z}$ be a family of independent ${\rm Sym}(m)$-valued random
variables. We assume here that (i) the mean of each random variable $q_n$ is
zero and (ii) there is a compact set $K\subset{\rm Sym}(m)$ so that the support
of  each $q_n$ is contained in $K$. By $\nu_n$ we denote the probability
measure of $q_n$. The probability measure for $q$ is then the product measure
$\nu = \otimes_{n\in\Z} \nu_n$. We use the notation $\E$ to denote the
expectation value with respect to this product measure, $\nu$.

On the Hilbert space $\ell^2(\Z;\C^m)$ (of $\C^m$-valued functions on $\Z$
equipped with the usual Euclidean norm) we consider the operator
\begin{equation}
H := \Delta + D + q\,,
\end{equation}
which is defined as
\begin{equation}\label{def:Hamiltonian}
(H \varphi)(n) := - \varphi(n-1) - \varphi(n+1) + D \,\varphi(n) + q_n\,
\varphi(n)\,, \quad \varphi\in \ell^2(\Z ; \C^m)\,,\; n \in \Z\,.
\end{equation}
To state the first result of this paper we introduce the following set which
depends on the (eigenvalues of the) constant ``potential'' $D$,
\begin{equation}\label{def:I_D}
I_D := \bigcap_{\lambda \in \sigma(D)} [\lambda - 2 , \lambda + 2 ]  \ .
\end{equation}
\begin{theorem} \label{thm:mainstrip} Let $\E[\sum_{n \in \Z} \| q_n \|^2 ] <
\infty$. Then almost surely $\sigma_{\rm ac}(H)\supseteq I_D$ and the spectrum
of $H$ is purely absolutely continuous in the interior of $I_D$.
\end{theorem}
Theorem \ref{thm:mainstrip} will be used to prove the second result of this
paper.
\begin{theorem} \label{thm:mainstripcor} Let $\E[\sum_{n \in \Z} \| q_n \|^2 ]
< \infty$. Then almost surely $\sigma_{\rm ac}(H)\supseteq \sigma(\Delta + D)$.
\end{theorem}

\vspace{0.5cm}
\noindent 
{\bf Remarks:} The case of a random potential with $m=1$ has been analyzed in
great detail by Delyon, Simon, and Souillard~\cite{DSS85}. For $m=1$ they not only prove
Theorem \ref{thm:mainstrip} (even under weaker conditions on the measures
$\nu_n$) but also that the rate of decay of $q_n$ is necessary in order to have
absolutely continuous spectrum. A deterministic version of a result in the direction of
Theorem  \ref{thm:mainstripcor} has been
announced by Molchanov and Vainberg~\cite{MV04} but, to  the best of our
knowledge, has not yet been published. Other previous work by Kirsch, Krishna, Obermeit and Sinha on decaying
potentials can be found in \cite{Kr}, \cite{KKO} and \cite{KrS}. We would also like to mention also the work of
Kotani and Simon~\cite{KS} and Schulz-Baldes~\cite{Schulz-Baldes} on random
Schr\"odinger operators on the strip.

\vspace{0.5cm}

\noindent
{\bf Example.} The most important application of Theorem \ref{thm:mainstripcor}
is to Schr\"odinger operators on a strip. More generally, let $\mathcal C :=
\{1,2,\ldots,L\}^d$ denote the discrete $d$-dimensional cube with side length
$L$. Then $\ell^2(\mathcal C)\cong \C^m$ with $m = L^d$ and 
$\ell^2(\Z;
\ell^2( \mathcal C) )\cong\ell^2(\Z\times\mathcal C)$.
We introduce the multi-index $\underline{n} :=
(n_1,n_2,\ldots,n_d) \in \Z^d$ with $|\underline{n}|:=|n_1|+|n_2|+\cdots+|n_d|$.
Let $D$ be the Dirichlet Laplacian on $\mathcal C$, i.e., for all
$\underline{n} \in \mathcal C$,
$$
(D \,\psi)(\underline{n}) := -\sum_{\underline{m}\in \mathcal C:
|\underline{m}-\underline{n}|=1}  \psi(\underline{m}) \,,\quad
\psi \in \ell^2(\mathcal C) \, .
$$
Note that $\Delta + D$ is 
equivalent to the
(nearest neighbor) Dirichlet Laplace operator on $\ell^2(\Z\times\mathcal C)$.
The eigenvalues of $D$ are indexed by $\underline{n} \in \mathcal C$ and are
given by $- 2\,\sum_{i=1}^d \cos (\pi n_i / (L+1))$. Observe that
$$
I_D = \big\{\lambda \in \R \ :  \ |\lambda | \leq 2 \big[ 1 - d \cos (\pi
/(L+1))\big]\big\} \,.
$$
If $d \geq 2$ this set is empty unless $L=1$. If $d=1$, then $I_D$ is non-empty
but its length converges to 0 as $L$ tends to infinity. By Theorem
\ref{thm:mainstripcor}, $\sigma_{\rm ac}(H) \supseteq \sigma(\Delta + D) =
[-2-2 d \cos(\pi/(L+1)), 2 + 2d \cos(\pi /(L+1))]$. By formally setting  $L$  to infinity the last interval becomes $[-2(d+1), 2 (d+1)]$.

\vspace{0.5cm}
\noindent
{\bf Remarks:} On the full two-dimensional lattice $\Z^2$, Bourgain~\cite{B02}
proved $\sigma_{\rm ac}(\Delta+q)\supseteq \sigma(\Delta)$ for Bernoulli and Gaussian distributed, independent random potentials whose variances decay faster than $|\underline{n}|^{-1/2}$. In~\cite{B03}, Bourgain improves this result to
the weaker $|\underline{n}|^{-1/3}$ decay rate.
For a deterministic potential, $q$, on $\Z^d$, Simon~\cite{Simon} conjectured
that if $(q_n/\sqrt{1+|\underline{n}|^{d-1}})_{n\in\Z^d}\in\ell^2(\Z^d)$, then
$\sigma_{\rm ac}(\Delta + q) = \sigma(\Delta)$. In dimension one this was
proved by Deift and Killip~\cite{DK99}. A recent improvement of this result
has been obtained by Denisov~\cite{D09}. In the analogous continuous setting,
progress has been made towards this $L^2$-conjecture e.g.\ by Denisov~\cite{D04}
and Laptev, Naboko, and Safronov~\cite{LNS}. For additional references see
\cite{CFKS}, \cite{Stoll}.

\section{Proofs of Theorem \ref{thm:mainstrip} and \ref{thm:mainstripcor}}

In order to prove the two main theorems in this paper we will study the Green's
functions defined by
\begin{equation}
G_n := P_n (H-\lambda)^{-1} P_n\,,\quad n\in\Z \,.
\end{equation}
Here, $P_n$ denotes the orthogonal projections of $\HH=\ell^2(\Z;\C^m)$ onto
the subspace $\ell^2(\{ n \}; \C^m)\cong \C^m$, and $\lambda$ denotes the
spectral parameter. Let $P_n^+:=\sum_{k\geq n} P_k$ and $P_n^-:=\sum_{k\leq n}
P_k$ be the orthogonal projections of $\HH$ onto the subspaces
$\ell^2(\{n,n+1,\ldots\};\C^m )$ and $\ell^2(\{\ldots,n-1,n \};\C^m )$,
respectively. Let
\begin{equation}
G_n^{\pm} := P_n \left( P_n^{\pm} (H_n - \lambda ) P_n^{\pm} \right)^{-1} P_n
\,,\quad n\in\Z
\end{equation}
be the so-called forward and backward Green's functions. Then we have the
recursion relation
\begin{equation} \label{eq:recurse0}
G_n =  -(G_{n+1}^+ + G_{n-1}^- +\lambda - D - q_n)^{-1} \,,\quad n\in\Z\,,
\end{equation}
which follows by using the decomposition $\HH = \ran P_n \oplus {\ran P_n}^
\perp$ and a resolvent identity.

If  ${\rm Im} \lambda > 0$, it is elementary to see that for each $n\in\Z$
\begin{equation}
G_n, G_n^{\pm} \in \SH_m := \{ Z=X + i Y : X, Y \in {\rm Sym}(m) , Y > 0 \} \,.
\end{equation}
Henceforth we will assume that ${\rm Im} \lambda > 0$. We equip the vector
space $\SH_m$ with the metric
$$
{\rm d}(Z,W) := \cosh^{-1} \Big( 1 + \frac{1}{2} {\rm cd}(Z,W) \Big)\,,
\quad Z , W \in \SH_m \,,
$$
where we have introduced
$$
{\rm cd}(Z,W) := {\rm tr}\left[({\rm Im} Z)^{-1} ( Z - W)^* ({\rm Im} W)^{-1}
(Z - W ) \right] \,.
$$
The space $(\SH_m,{\rm d})$ is called Siegel half space and is a generalization
of the usual Poincar\'e upper half plane.

By symmetry it will suffice to study $G_0^+$. Using the decomposition
$\ran P_0^+ = \ran P_0 \oplus \ran P_1^+$, we see that
\begin{equation} \label{eq:recursion}
G_0^+ = \Phi_{q_0}(G_1^+ ) \,,
\end{equation}
with the following mapping on $\SH_m$
$$
\Phi_{\delta}(Z)  :=  - ( Z + \lambda - D - \delta )^{-1}\,,\quad \delta \in
{\rm Sym}(m)\,,Z \in \SH_m\,.
$$
Iterating Eq. \eqref{eq:recursion} we arrive at
\begin{equation} \label{eq:iterate}
G_0^+ = \Phi_{q_0} \circ  \Phi_{q_1} \cdots  \circ\Phi_{q_n}(G_{n+1}^+) \,,
\quad n\in\N_0\,.
\end{equation}
We will use the following theorem, which is a special case of a theorem
obtained in \cite{FHS07}.
\begin{theorem}  \label{thm:fhs1}  Let ${\rm Im} \lambda > 0$ and
$(\Lambda_n )_{n \in \N_0} \subset \SH_m$ be any sequence.  Then
$$
G_0^+  = \lim_{n \to \infty} \Phi_{q_0} \circ \cdots \circ
\Phi_{q_n}(\Lambda_n) \,.
$$
\end{theorem}
We present a direct proof of this theorem in Appendix A, 
which in this case is simpler than the proof given in \cite{FHS07} for more
general graphs.

\vspace{0.5cm}

The next theorem measures the distance of $G_0^+$ from the free forward
Green's function, which is determined by the following fixed point relation
in $\SH_m$:
\begin{equation}
Z_\lambda = \Phi_0(Z_\lambda)\,.
\end{equation}
Solving for $Z_\lambda$ yields
\begin{equation}
Z_\lambda = \frac{D - \lambda}{2}  + i \sqrt{ 1 - \left(\frac{D - \lambda}{2}
\right)^2 } \,.
\end{equation}
Note that for real $\lambda$, we have ${\rm Im} Z_\lambda > 0$ if and only if
$\lambda$ is in the interior of $I_D$.
To formulate the next theorem we define 
\begin{equation}
{\rm cd}_\lambda(Z) := {\rm cd}(Z_\lambda , Z )\,,\quad Z\in \SH_m \,.
\end{equation}

\begin{theorem}  \label{thm:mainstrip0} Suppose $\E[\sum_{n \in \Z}
\|q_n\|^2 ] < \infty$. Let $J$ be a closed subset of the interior of $I_D$.
Then
\begin{equation}
\sup_{\lambda \in J + i(0,1]} \E\left[{\rm cd}_\lambda^2(G_n^{\pm})\right]
< \infty \,.
\end{equation}
\end{theorem}

\begin{proof} By symmetry it suffices to consider  without loss of generality
$G_0^+$. We assume $\lambda \in J + i (0, 1]$ is fixed. By Theorem
\ref{thm:fhs13}, we know that $G_0^+ = \lim_{n \to \infty} Z_{0,n}$, where
$Z_{0,n} = \Phi_{q_0} \circ \Phi_{q_1} \circ \cdots \circ \Phi_{q_n}
(Z_\lambda)$. Moreover, by Lemma \ref{thm:fhs13} from Appendix B, we know that there exists a
hyperbolic ball $B \subset \SH_m$
such that $Z_{0,n} \in B$ for  all $n\geq 2$ and potentials $q$ with
$q_k \in K^{}$. By continuity of the function $Z \mapsto {\rm cd}^2_\lambda(Z)$,
we have
$$
\lim_{n \to \infty} {\rm cd}_\lambda^2(Z_{0,n}) = {\rm cd}^2_\lambda(G_0^+)\,.
$$
Since ${\rm cd}_\lambda^2(Z)$ is bounded on the ball $B$, it follows from
dominated convergence that
$$
\E[{\rm cd}_\lambda^2(G_0^+) ] = \lim_{n \to \infty} \E[{\rm cd}_\lambda^2
(Z_{0,n})] \,.
$$
It remains to show that the right-hand side is bounded uniformly in
$\lambda \in J + i (0, 1]$.
To this end we set $Z_{\ell,n} := \Phi_{q_\ell} \circ \Phi_{q_{\ell+1}} \circ
\cdots \circ \Phi_{q_n}(Z_\lambda)$. Note that $Z_{\ell,n} = \Phi_{q_\ell}
(Z_{\ell+1,n})$. Using the  inequality of Lemma \ref{lem:2} below, we find
\begin{eqnarray*}\lefteqn{
\E[{\rm cd}_\lambda^2(Z_{0,n})] + 1}
\\
&=& \int_{K^{n+1}} ({\rm cd}_\lambda^2(Z_{0,n}) + 1 ) \, d \nu_0(q_0)
\cdots d \nu_n(q_n)
\\
&=& \int_{K^{n+1}} \frac{{\rm cd}_\lambda^2[\Phi_{q_0} ( Z_{1,n})] + 1 }{{\rm
cd}_\lambda^2(Z_{1,n}) + 1} \,({\rm cd}_\lambda^2(Z_{1,n}) + 1 ) \,
d \nu_0(q_0) \cdots d \nu_n(q_n)
\\
&\leq& \int_K ( 1 + A(Z_{1,n},q_0)  + C_0 \|q_0\|^2 )\, d\nu_0(q_0)
\int_{K^n} ({{\rm cd}_\lambda^2(Z_{1,n}) + 1}) \,
d \nu_1(q_1) \cdots d \nu_n(q_n)
\\
&=& ( 1 + C_0 \,\E[\|q_0\|^2] ) \int_{K^n} ({\rm cd}_\lambda^2(Z_{1,n}) + 1) \,
d \nu_1(q_2) \cdots d \nu_n(q_n)
\\
&\vdots& \\
&\leq& \prod_{i=0}^n ( 1 + C_0 \,\E[\|q_i\|^2] )
\\
&\leq&  \exp  (C_0 \sum_{i=0}^\infty \E[\|q_i\|^2] ) \,,
\end{eqnarray*}
where we have used $\int A(z, q) \,d\nu_i(q) = 0$, which follows from the
assumption that $q_i$ is a random variable with mean zero.
\end{proof}

\begin{lemma}
\label{lem:2} Suppose $K$ is a compact subset of $\SH_m$. Let $J$ be a closed interval contained in the
interior of $I_D$. Then there exists a constant $C_0$ and a linear functional
$A(Z, \cdot ) : {\rm Sym}(m) \to \R$, depending continuously on $Z \in \SH_m$,
such that for all  $\lambda \in J + i (0,1]$ and $Z \in \SH_m$,
\begin{align} \label{eq:importantestimate}
\frac{{\rm cd}_\lambda^2(\Phi_\delta(Z)) + 1 }{{\rm cd}_\lambda^2(Z) + 1 }
\leq  1 + A(Z,\delta) + C_0 \| \delta \|^2\,,   \quad  \forall \delta \in K^{} \,.
\end{align}
\end{lemma}

\begin{proof} Using that $\Phi_0$ is a hyperbolic contraction, we see ${\rm
cd}_\lambda(\Phi_\delta(Z))={\rm cd}(\Phi_\delta(Z),Z_\lambda) \leq
{\rm cd}(Z - \delta, Z_\lambda) = {\rm cd}_\lambda(Z-\delta)$.
By the definition of the distance function we have
\begin{align*}
{\rm cd}(Z - \delta, Z_\lambda) = {\rm cd}(Z , Z_\lambda) + a(Z,\delta) +
b(Z,\delta) \,,
\end{align*}
where (with $Z_{\lambda} = X_{\lambda} + i Y_{\lambda}$ and $Z = X + i Y$),
\begin{align*}
a(Z,\delta) &:= - {\rm tr}\big[Y_\lambda^{-1/2}\delta Y^{-1}(Z  - Z_\lambda)
Y_\lambda^{-1/2} \big] - {\rm tr}\big[ Y_\lambda^{-1/2}(Z  - Z_\lambda)^*Y^{-1}
\delta Y_\lambda^{-1/2}\big]\,,
\\
b(Z,\delta) &:= {\rm tr}( Y_\lambda^{-1/2} \delta  Y^{-1} \delta
Y_\lambda^{-1/2} ) \,.
\end{align*}
Using the Cauchy-Schwarz  inequality it follows that
$$
{\rm L. \ H. \ S. \ of } \ \eqref{eq:importantestimate} \leq
1 + A(Z,\delta) + C(Z,\delta) \,,
$$
with
\begin{align*}
 A(Z,\delta) &:= \frac{2\, {\rm cd}_\lambda(Z) \,a(Z,\delta)}{{\rm cd}_\lambda^2(Z) + 1} \ ,
 \\
 C(Z,\delta) &:= \frac{2\, a(Z,\delta)^2 + 2\, {\rm cd}_\lambda(Z) \,b(Z,\delta) + 2
\, b(Z,\delta)^2 }{{\rm cd}_\lambda^2(Z) + 1 } \,.
\end{align*}
It remains to show that $C(Z,\delta) \leq C_0 \| \delta \|^2$ for some $C_0$.
Let us use the bounds,
\begin{align}
\label{a-inequal}
a(Z,\delta)^2 &\leq 4 \,{\rm cd}_\lambda(Z) \,b(Z,\delta) \,,
\\
\label{b-inequal} b(Z,\delta) &\leq \| Y_\lambda^{-1} \|^2\, \| \delta \|^2 \,
{\rm tr} ( Y_\lambda^{1/2}  Y^{-1}  Y_\lambda^{1/2} ) \,.
\end{align}
(\ref{a-inequal}) follows from the Cauchy-Schwarz inequality. The
trace in the function $b$ can be written as ${\rm tr}(EFE)$ with $E:=Y_\lambda^{-1/2}\delta Y_\lambda^{-1/2}$ and $F:=Y_\lambda^{1/2}Y^{-1}
Y_\lambda^{1/2}$. This trace is estimated from above by $\|E^2\| \,{\rm tr} F$. Then
use \linebreak  $\|E^2\|\leq\|Y_\lambda^{-1/2}\|^2 \,\|Y_\lambda^{-1}\|\, \|\delta\|^2$.
Since $Y_\lambda$ is self-adjoint $\|Y_\lambda^{-1/2}\|^2 =\|Y_\lambda^{-1}\|$, and
(\ref{b-inequal}) follows.

The next estimate allows us to bound the right-hand side of  \eqref{b-inequal} in terms of  ${\rm
cd}_\lambda(Z)$.
\begin{align}\lefteqn{
{\rm tr} (Y_\lambda^{1/2} Y^{-1} Y_\lambda^{1/2})}\nonumber
\\
& \leq {\rm tr} (Y_\lambda^{1/2} Y^{-1} Y_\lambda^{1/2})  + {\rm tr} (  Y^{1/2}    Y_\lambda^{-1} Y^{1/2}       )  \nonumber   \\
& = {\rm tr}\big[Y_\lambda^{-1/2} ( Y - Y_\lambda)   Y^{-1} ( Y - Y_\lambda)
Y_\lambda^{-1/2} \big] + 2 m  \nonumber  \\
&\leq {\rm tr}\big[Y_\lambda^{-1/2} ( Y - Y_\lambda)   Y^{-1} ( Y - Y_\lambda)
 Y_\lambda^{-1/2} \big]
+ {\rm tr}\big[Y_\lambda^{-1/2}  (X - X_\lambda)  Y^{-1} ( X - X_\lambda)
Y_\lambda^{-1/2}\big] + 2 m  \nonumber \\
& = {\rm cd}_\lambda(Z) + 2 m  \,. \label{eq:traceineq}
\end{align}
The claim now follows by inserting the above estimates and using that $\| Y_\lambda \|$ and $\| Y_\lambda^{-1} \|$ are uniformly bounded for
$\lambda \in J + i(0,1]$ and that $\delta$ is contained in a bounded set.
\end{proof}

\vspace{0.5cm}
\noindent
{\it Proof of Theorem \ref{thm:mainstrip}.}

\noindent
\underline{Step 1:} Almost surely $\sigma(H) \supseteq \sigma(\Delta + D)$.

\vspace{0.5cm}

The condition $\mathbb{E}[\sum_{n \in \Z} \| q_n \|^2 ]$ implies that almost
all potentials are in $\ell^2$ and thus decay at infinity. $H$ is thus a compact
perturbation of $\Delta + D$ and hence $\sigma(\Delta + D)=\sigma_{{\rm ess}}
(\Delta + D) = \sigma_{{\rm ess}}(H)\subseteq\sigma(H)$ by Weyl's Theorem.


\vspace{0.5cm}
\noindent
\underline{Step 2:}   Let $J$ be any closed interval contained in the interior
of $I_D$. Let $W_\lambda := - \left(2 Z_\lambda + \lambda - D\right)^{-1}$.
Then
$$\sup_{\lambda \in J + i (0,1]} \E \left[{\rm cd}^2(G_n,W_\lambda)
\right] < \infty\,.
$$

\vspace{0.5cm}
If we use the recursion relation \eqref{eq:recurse0}, the fact that
$Z \mapsto - Z^{-1}$ is a hyperbolic isometry, and the inequalities of Lemma
\ref{lem:ineqappendix} (given in the Appendix B) we find that
\begin{align*}
{\rm cd}(G_n,W_\lambda) &\leq {\rm cd}(G_{n+1}^+ + G_{n-1}^{-} + q_n , 2
Z_\lambda )
\\
&\leq {\rm cd}(G_{n+1}^+ + q_n/2, Z_\lambda) +  {\rm cd}(G_{n-1}^- + q_n/2,
Z_\lambda)
\\
&\leq C \left[  1 + {\rm cd}(G_{n+1}^+, Z_\lambda) +
{\rm cd}(G_{n-1}^-, Z_\lambda)  \right]  (1 + \| q_n \|^2 )\,.
\end{align*}
Then,
\begin{align*}
\E \left[{\rm cd}^2(G_n,W_\lambda)\right]
&\leq C \big(1+ {\rm cd}_\lambda(G_{n+1}^+) + {\rm cd}_\lambda(G_{n+1}^-)\big)
\,,
\end{align*}
and Step 2 follows from Theorem \ref{thm:mainstrip0}.

\vspace{0.5cm}
\noindent
\underline{Step 3:} Almost surely $H$ has purely ac spectrum in the interior of  $I_D$.

\vspace{0.5cm}

For $x \in \Z \times \{ 1,...,m\}$ let $\mu_x$ denote the spectral measure of
$H$ for the indicator function at $x$, $1_x \in \HH \cong \ell^2(\Z \times
\{ 1,...,m\})$.
Step 2 implies that almost surely $\mu_x$ is absolutely continuous on any closed
subset of the interior of $I_D$. This  can be seen for example  by applying
Lemma 1 in \cite{FHS09} and noting that for any closed subset $J$ contained
in the interior of $I_D$ there exists a constant $C$ such that (see Lemma
\ref{lem:ineqappendix}) ${\rm tr} \, ({\rm Im} Z) \leq C \,( {\rm cd}_\lambda(Z)
+1)$ for all $\lambda \in J$ and $Z\in\SH_m$; see also \cite[Theorem 4.1]{K98}.
Now choosing
a sequence of closed subsets $(J_n)_{n \in \N}$ of the interior of $I_D$, such
that $J_n \subset J_{n+1}$ and $\bigcup_{n=1}^\infty J_n = I_D$, and using
that countable unions of sets of measure zero have again measure zero, we find
almost surely that for all $x \in \Z \times \{ 1,...,m \}$ the spectral measure
$\mu_x$ is absolutely continuous on the interior of  $I_D$.

\vspace{0.5cm}

\noindent
The theorem now  follows by combining Steps 1 and 3. \qed

\vspace{0.5cm}

To prove Theorem \ref{thm:mainstripcor} we  will use Theorem \ref{thm:mainstrip}
in combination with Theorem \ref{thm:denisov} below. Theorem \ref{thm:denisov}
is an extension
of a theorem by Denisov~\cite[Theorem 1.2]{D06}. A proof can also be found in
Albeverio and Konstantinov
\cite{AK08}. We give a proof in Appendix C following arguments given in
\cite{AMM03,D06}.

\begin{theorem}[Denisov] \label{thm:denisov}
Let $H_1$ and $H_2$ be two bounded self-adjoint operators on the Hilbert spaces
$\HH_1$ and $\HH_2$, respectively. Assume that for $a<b$, $[a,b] \subseteq
\sigma_{\rm ac}(H_1)$ and $\sigma_{\rm ess}(H_2) \subseteq (-\infty , a ]
\cup [b , \infty)$.
Let $V : \HH_2 \to \HH_1$ be a Hilbert-Schmidt operator (i.e., $V^*V$ and
$V V^*$ are trace class operators on $\HH_2$ respectively $\HH_1$) and let
$H_V := \left[ \begin{array}{cc}H_1 &  V  \\ V^*  &  H_2
\end{array} \right]$. Then, $[a,b] \subseteq \sigma_{\rm ac}(H_V)$.
\end{theorem}

\vspace{0.5cm}
\noindent
{\it Proof of Theorem \ref{thm:mainstripcor}.}
Let $\{\mu_1,\ldots,\mu_m\}$ be the eigenvalues of $D$ and let $\lambda\in \mathbb R$.
Then the eigenvalues of $Z_\lambda$ are given by
$z_{\lambda, k}=(\mu_k-\lambda)/2 + i\sqrt{1-((\mu_k-\lambda)/2)^2}$. If
$\lambda\in [\mu_k-2,\mu_k+2]$ then $z_{\lambda, k}$ lies on the unit semicircle above the real axis. Otherwise $z_{\lambda, k}$ lies on the real axis outside the unit circle (see diagram)

\begin{center}\includegraphics{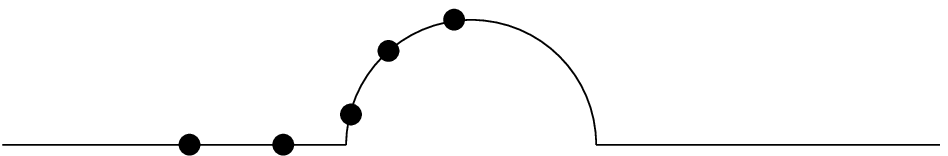}\end{center}

Since $Z_\lambda$ is related to the Green's function for $\Delta + D$, a point $\lambda$ lies in $\sigma(\Delta + D)$ if and only if at least one of the $z_{\lambda, k}$ lies on the semicircle, and thus has positive imaginary part. Let $m(\lambda)$ denote the number of $z_{\lambda, k}$ on the semicircle. As we vary $\lambda$, the function $m(\lambda)$ is locally constant, with jumps when one of the $z_{\lambda, k}$ moves in or out of the semicircle.
Pick a  $\lambda_0$ and let $I$ be the largest interval containing $\lambda_0$ on which $m(\lambda)$ is constant. Notice that $\sigma(\Delta + D)$ is a finite disjoint union of such intervals.
The collection of $z_{\lambda, k}$ that remain in the semicircle for $\lambda \in I$ corresponds to a subset of eigenvalues of $D$, and thus to a spectral projection $P_I$ on $\mathbb C^m$.
We identify the range of $P_I$ with $\C^{m(\lambda)}$. We use the same notation for the projection on $\HH=\ell^2(\Z;\C^m)$ where $P_I$ acts as a (constant) multiplication operator. Introducing $\overline{P}_I := 1 - P_I$ we have the decomposition
\ben 
H = \Delta + D + q = \left[ \begin{array}{cc}
P_I (\Delta + D + q) P_I & P_I q \overline{P}_I  \\
\overline{P}_I  q  P_I &  \overline{P}_I (\Delta + D + q) \overline{P}_I
\end{array} \right] \,.
\een
Note that $\Delta_I:=P_I \,\Delta \,P_I$ is just the Laplace operator (\ref{def:Hamiltonian}) on $\ell^2(\Z;\C^{m(\lambda)})$. Furthermore, let $D_I$
be the restriction of $D$ onto $\C^{m(\lambda)}$. By Theorem~\ref{thm:mainstrip},
$P_I (\Delta + D + q) P_I = \Delta_I +  D_I +  P_I q P_I $ has almost surely ac spectrum on $I$, since $I \subset I_{D_I}$, see \eqref{def:I_D}. Since almost surely $q$ is in $\ell^2$ and thus decays at infinity the essential spectrum of  $\overline{P}_I (\Delta + D + q)
\overline{P}_I$ is contained in the complement of the interior of $I$.
Since $P_I q \overline{P}_I$ is Hilbert-Schmidt almost surely, we can apply Theorem \ref{thm:denisov} and hence conclude that almost surely  $I \subseteq \sigma_{\rm ac}(H + D + q)$. Repeating the above arguments for the remaining intervals of non-zero length in the decomposition of the spectrum of $\Delta + D$ yields the claim.\qed

\section*{Acknowledgement}

W.S. wants to thank the University of British Columbia for hospitality
and financial support. D. H. wants to acknowledge the summer research grant awarded by
the College of William \& Mary.

\section*{Appendix A: Proof of Theorem \ref{thm:fhs1}} 

Let us start with the following lemma.
\begin{lemma} \label{thm:fhs13} Suppose that
$|\lambda|, \| \delta_1 \| , \|\delta_2 \| \leq C$ and ${\rm Im \lambda}
\geq 1/C$. Then there exists a compact set $B \subset \SH_m$ (depending on
$C$) such that $\Phi_{\delta_1} \circ \Phi_{\delta_2}(\SH_m) \subseteq B$.
\end{lemma}
\begin{proof}
Applying $\Phi_{\delta}$ once yields an upper bound on the norm, which can be
seen from the basic inequality
$$
\| ( Z + \lambda - D - \delta)^{-1} \| \leq ( {\rm Im} \lambda)^{-1}\,,\quad
Z\in\SH_m\,.
$$
Applying $\Phi_{\delta}$ a second time yields a lower bound on the imaginary
part, which can be seen by the following estimate 
$$ {\rm Im}\big[-(Z+\lambda - D -\delta)^{-1}\big]
\,=\, (Z^*+\lambda^* - D -\delta)^{-1} \,{\rm Im}(Z+\lambda)\,
  (Z+\lambda - D -\delta)^{-1}\,
\geq \,\frac{{\rm Im} \lambda}{\| Z + \lambda - D - \delta \|^2} \,.
$$
\end{proof}

\vspace{0.5cm}
\noindent {\it Proof of Theorem \ref{thm:fhs1}.} \\
\noindent
\underline{Step 1:} Let $B$ be a compact subset of $\SH_m$ as in the previous
Lemma \ref{thm:fhs13}. Then there exists a $\gamma < 1$, such that for all
$Z,W \in B$,
$$ {\rm d}(\Phi_\delta(Z), \Phi_\delta(W)) \leq \gamma \,{\rm d}(Z,W)\,.
$$
\vspace{0.5cm}

Using that the maps $Z \mapsto - Z^{-1}$ and $Z \mapsto Z - D - \delta$ are
hyperbolic isometries on $\SH_m$, we find that
$$
{\rm d}(\Phi_\delta(Z),\Phi_\delta(W)) = {\rm d}(Z + \lambda, W + \lambda)\,.
$$
In order to estimate the last expression we use that
$$ [{\rm Im}(W+\lambda)]^{-1} = [{\rm Im}(W)]^{-1/2}\,
\underbrace{[{\rm Im}(W)]^{1/2}\,[{\rm Im}(W+\lambda)]^{-1}\,
[{\rm Im}(W)]^{1/2}}_{\leq \sqrt{\gamma}}\,[{\rm Im}(W)]^{-1/2}
\leq \sqrt{\gamma} \;[{\rm Im}(W)]^{-1}
$$
for some real number $\gamma < 1$ since $W$ is in a bounded set. If we apply
this estimate also to $Z$ we obtain that
$$ {\rm d}(\Phi_\delta(Z),\Phi_\delta(W))\, = \,
 {\rm d}(Z + \lambda  , W + \lambda ) \,\leq \,\gamma \,{\rm d}(Z,W) \,.
$$

\vspace{0.5cm}

\noindent
\underline{Step 2:} The sequence $(\Phi_{q_0} \circ \cdots \circ \Phi_{q_n}
(\Lambda_n))_{n \in \N_0}$ converges to a limit independent of the choice of
$(\Lambda_n )_{n \in \N_0}$.

\vspace{0.5cm}

Suppose $(\widetilde{\Lambda}_n )_{n \in \N_0}$ is a different sequence. Then
\begin{equation} \label{eq:cauchy}
{\rm d}( \Phi_{q_0} \circ \cdots \circ \Phi_{q_n}(\Lambda_n) , \Phi_{q_0}
\circ \cdots \circ \Phi_{q_n}(\widetilde{\Lambda}_n) )
\leq \gamma^{n-2} C \to 0 \,, \quad ( n \to \infty ) ,
\end{equation}
with  $C := \sup_{(Z,W)\in B^2}{\rm d}(Z,W)$. We conclude that if the limit
exists it must be independent of the sequence $(\Lambda_n )_{n \in \N_0}$. On
the other hand the sequence $(\Phi_{q_0} \circ \cdots \circ \Phi_{q_n}
(\Lambda_n))_{n \in \N_0}$ is a Cauchy sequence, which can be seen by
inserting $\widetilde{\Lambda}_n :=  \Phi_{q_{n+1}} \circ \cdots \circ
\Phi_{q_{n+m}}(\Lambda_{n+m})$ for $m\in\N$ into \eqref{eq:cauchy}.

\vspace{0.5cm}

\noindent
The theorem now follows from Step 2 and Eq. \eqref{eq:iterate}.
\qed

\section*{Appendix B: Some inequalities}

\begin{lemma} \label{lem:ineqappendix}
Let $Z_i \in \mathbb{SH}_m, i\in\{0,1,2\}$ and $\delta \in {\rm Sym}(m)$. Then
\begin{itemize}
\item[(a)] ${\rm cd}(2 Z_0, Z_1 + Z_2 ) \leq \frac12 \big[{\rm cd}(Z_0, Z_1) +
{\rm cd}(Z_0, Z_2)\big]$.
\item[(b)] ${\rm cd}(Z_0, \delta + Z_1 ) \leq C  \,( 1 + \| \delta \|^2) \,
\big[{\rm cd}(Z_0, Z_1) + 1\big]$ for some constant $C$ that depends on the
(norm of the) imaginary part of $Z_0$.
\item[(c)] For $\lambda\in I_D$ (see definition (\ref{def:I_D})) there is a constant $C$ (depending on $\lambda$ and $m$)
so that ${\rm tr} \, ({\rm Im} Z_0) \leq C \,( {\rm cd}_\lambda(Z_0) + 1 )$.
\end{itemize}
\end{lemma}

\begin{proof}
(a). Let us define
$$ A:=\left[\begin{array}{c} U_1\\U_2\end{array}\right]\,,\quad
B:=\left[\begin{array}{c}Y_1^{1/2}(Y_1+Y_2)^{-1/2}\\Y_2^{1/2}(Y_1+Y_2)^{-1/2}
\end{array}\right]
$$
with $Y_i:={\rm Im}(Z_i)$ and
$U_i = Y_i^{-1/2}(Z_0-Z_i)Y_0^{-1/2}$.
Then 
\begin{eqnarray*}{\rm cd}(2Z_0, Z_1 + Z_2 )
= \frac12 {\rm tr }\Big[(U_1^*Y_1^{1/2} + U_2^*Y_2^{1/2})(Y_1+Y_2)^{-1}
(Y_1^{1/2}U_1+Y_2^{1/2}U_2)\Big]
= \frac12 {\rm tr }\big[A^*BB^*A\big]\,.
\end{eqnarray*}
Since $B^*B=1$, $BB^*$ is a projection and hence $BB^*\le 1$. Therefore,
${\rm tr }\big[A^*BB^*A\big]\le {\rm tr }\big[A^*A\big]={\rm tr }\big[U_1^*U_1
\big] + {\rm tr }\big[U_2^*U_2\big] = {\rm cd}(Z_0, Z_1) + {\rm cd}(Z_0, Z_2)$.
\\

(b). By expanding the product in the trace and using ${\rm tr}[A^* B] \leq
({\rm tr}|A|^2)^{1/2}  ({\rm tr} |B|^2)^{1/2} \leq \frac{1}{2}{\rm tr} |A|^2 +
\frac{1}{2}{\rm tr} |B|^2$ we obtain
\begin{align*}
{\rm cd}(Z_0, \delta + Z_1 ) &\leq 2\, {\rm cd}(Z_0, Z_1 ) + 2 \,{\rm tr}
( Y_0^{-1/2} \delta Y_1^{-1} \delta Y_0^{-1/2} )
\\
&\leq  2 \,{\rm cd}(Z_0, Z_1) + 2 \,\| Y_0^{-1/2}\delta Y_0^{-1/2}\|^2 \,
{\rm tr}(Y_0^{1/2} Y_1^{-1} Y_0^{1/2} )\,.
\end{align*}
Now we use $\| Y_0^{-1/2}\delta Y_0^{-1/2}\|\leq \|\delta\|\,
\|Y_0^{-1/2}\|^2$ and inequality \eqref{eq:traceineq}, i.e.,
${\rm tr}(Y_0^{1/2} Y^{-1} Y_0^{-1/2} ) \leq {\rm cd}(Z_0,Z_1) + 2 m$, which,
all put together, proves the claimed inequality.
\\

(c). We have ${\rm cd}_\lambda(Z_0) \geq {\rm tr}\big[Y_\lambda^{-1/2}
(Y_0-Y_\lambda) Y_0^{-1}(Y_0-Y_\lambda)Y_\lambda^{-1/2}\big]
\geq {\rm tr}\big[Y_\lambda^{-1/2} Y_0 Y_\lambda^{-1/2}\big] - 2m
\geq C'\, {\rm tr}(Y_0) - 2m$,
where the constant $C'$ depends on $\lambda$ throught the estimate on
$Y_\lambda^{-1}$. Finally, we choose the constant $C:=2m/C'$ and the stated
inequality follows.
\end{proof}

\section*{Appendix C: Proof of Theorem \ref{thm:denisov}}

To prove Theorem \ref{thm:denisov} we will use the following theorem by
Albeverio, Makarov, and Motovilov~\cite{AMM03}. For the convenience of the
reader we also present a proof that follows closely the one given in
\cite{AMM03} but uses analytic perturbation theory to obtain the graph
subspaces.

\begin{theorem}[Albeverio-Makarov-Motovilov]  \label{thm:albeverio}
Let $H_1$, $H_2$ be two bounded self-adjoint operators on the Hilbert spaces
$\HH_1$ and $\HH_2$, respectively. Assume that $\sigma(H_1) \subset (a, b)
\subseteq \rho(H_2)$ and $V:\HH_2 \to \HH_1$ is a Hilbert-Schmidt operator.
Let
\begin{equation} H_0 := \left[ \begin{array}{cc} H_1 &  0  \\ 0  &  H_2
\end{array} \right]\,,\quad W:=\left[ \begin{array}{cc} 0 &  V  \\ V^*  &  0
\end{array} \right]\,,\quad H_V:=H_0+W\,.
\end{equation}
Then $\sigma_{\rm ac}(H_V) = \sigma_{\rm ac}(H_0)$.
\end{theorem}

\begin{proof}
Since finite rank perturbations $F$ do not change the ac spectrum and
such operators are norm-dense in the space of Hilbert-Schmidt operators
we can replace $V$ by $V+F$ and achieve that its norm is small. Henceforth we
assume that $\| V \|$ is small. Let $\HH := \HH_1 \oplus \HH_2$, and for
$i=1,2$ we denote by $p_i$ the orthogonal projection of $\HH$ onto
$\HH_i$.

\vspace{0.5cm}

\noindent
\underline{Step 1:}  For $\|V\|$ sufficiently small, there exist orthogonal
projections $P_1$ and $P_2$ such that $P_1 + P_2 = 1$, $P_i H  = H P_i$, and
$p_i P_i p_i : \HH_i \to \HH_i$ is bijective, for $i\in\{1,2\}$. Furthermore,
$P_i - p_i$ is Hilbert-Schmidt and its norm can be made arbitrarily small
if we choose $\|V\|$ small enough.

\vspace{0.5cm}

Let $\Gamma_1$ [$\Gamma_2$] be  a counter-clockwise contour in $\C$ around the spectrum of $H_1$ [$H_2$] and contained in the resolvent set of $H_{2}$
[$H_1$]. Then we define the spectral projections
\begin{equation}
P_i := \frac{1}{2\pi i} \int_{\Gamma_i} \frac{dz}{z - H_V}\,,\quad i\in\{1,2\}\,.
\end{equation}
The first two properties follow from this representation.
In order to verify the other statements we use a similar representation for
the projections $p_i$ and the resolvent identity so that
$$ P_i - p_i = \frac{1}{2\pi i} \int_{\Gamma_i} (z - H_0)^{-1}
W (z - H_V)^{-1}\,{dz}\,.
$$
Hence, $P_i - p_i$ is Hilbert-Schmidt. If $\|V\|$ is small then
$\|p_iP_ip_i - p_i\|$ is small, too, and therefore $p_iP_ip_i$ can be inverted
on $\HH_i$.

\vspace{0.5cm}

\noindent
\underline{Step 2:}
For $\| V \|$ sufficiently small,  there exist  operators $Q_1 : \HH_1 \to
\HH_2$ and  $Q_2 : \HH_2 \to
\HH_1$  such that $\ran P_1 = \{(x,Q_1x) | x \in \HH_1\}$ and
$\ran P_2 = \{(Q_2x,x) | x \in \HH_2\}$. Moreover  $Q_2 =  -Q_1^*$ and  $Q_i$, for $i\in\{1,2\}$, is
Hilbert-Schmidt and its norm can be chosen arbitrarily small for  $\|V\|$
sufficiently small.

\vspace{0.5cm}

Let $i \neq j$. First observe that  the operator $p_j P_i = p_j (P_i - p_i)$ as well as its
adjoint are Hilbert-Schmidt  and can be made arbitrarily small by Step 1. Using the identity
$P_1 p_1 + P_2 p_2 = 1 - P_1 p_2 - P_2 p_1$ and noting that the r.h.s. can be made arbitrarily close to
one, we see that $\HH = \ran P_1 p_1 + \ran P_2 p_2$. This implies that $\ran P_i = \ran P_i p_i$.
Define $Q_i := p_j P_i (p_i P_i p_i)^{-1}$.  If we set
$x:=p_i (P_i z)$ for $z\in\HH_i$, then $p_j (P_i  z) = Q_i x$. Hence, the range of
$P_i p_i$ equals the graph of $Q_i$. The statement $Q_2 =  -Q_1^*$ follows by orthogonality.

\vspace{0.5cm}

\noindent
\underline{Step 3:} For $\|V\|$ sufficiently small $H_V$ is unitary equivalent
to
$$
\left[ \begin{array}{cc}
H_1  + T_1 &  0  \\
0  &  H_2 + T_2
\end{array} \right] \, ,
$$
where $T_i$ are trace class operators.

\vspace{0.5cm}

Since by Step 2, $H_V$ leaves the graphs of the operators $Q_1$ and $Q_2$
invariant, there exist operators $A_i \in \mathcal{\HH}_i$ such that
\begin{equation}\label{eq:Qdecop}
H_V ( 1 + Q) = (1 + Q ) A \ ,
\end{equation}
with
$$
Q := \left[ \begin{array}{cc}
0  &  Q_2  \\
Q_1  &  0
\end{array} \right] \quad {\rm and } \quad
A := \left[ \begin{array}{cc}
A_1  &  0  \\
0  &  A_2
\end{array} \right] \ .
$$
To construct these operators, let $x\in\HH_1$. Then by Step 2, $z:=(x,Q_1x)\in \ran P_1$. By
Step 1, $H_Vz\in \ran P_1$, and again by Step 2, $H_Vz=(y,Q_1y)$ for some
uniquely determined $y\in\HH_1$. This defines the operator $A_1:\HH_1\to\HH_1$
by setting $A_1x:=y$. A similar construction gives the operator $A_2$. As
a result we obtain (\ref{eq:Qdecop}). Expanding the product in (\ref{eq:Qdecop})
we see more concretely that $A_1 = H_1 + V Q_1$ and $A_2 = H_2 +  V^* Q_2$.
Since $Q$ has purely imaginary spectrum the operator  $1 + Q$ is bijective.
Using the polar decomposition $1 + Q = U |1+Q|$, with  $U$ unitary, we find
\begin{equation}\label{eq:UHU}
U^* H_V U =  | 1 + Q |  \;  A  \; |1 + Q|^{-1} \ .
\end{equation}
Note that
$$
|1 + Q| = \left[ \begin{array}{cc}(1+Q^*_1 Q_1)^{1/2} &  0  \\
0  & (1+ Q_2^* Q_2)^{1/2}\end{array} \right] \,.
$$
Using $0\le (1+Q^*_1 Q_1)^{1/2} - 1\le Q^*_1 Q_1$, the operator
$(1+Q^*_1 Q_1)^{1/2} \,(H_1 + V Q_1)\,(1+Q^*_1 Q_1)^{-1/2}-H_1$ is trace class.
A similar statement holds for the second diagonal operator on the right-hand
side of (\ref{eq:UHU}).
\vspace{0.5cm}

\noindent
\underline{Step 4:} $\sigma_{\rm ac}(H) = \sigma_{\rm ac}(H_0)$.

\vspace{0.5cm}

This follows from the result of Step 3 and the fact the trace class
perturbations preserve ac spectrum.

\end{proof}

Our proof of  Theorem \ref{thm:denisov} follows closely the one given by
Denisov~\cite{D06}, but uses almost analytic functional calculus
(cf.~\cite{Da95}) to control the function of an operator.

\vspace{0.5cm}

\noindent
{\it Proof of Theorem  \ref{thm:denisov}.}  Fix  $\epsilon > 0$.  We will
show that $[a+\epsilon,b-\epsilon ] \subset \sigma_{\rm ac}(H_V)$. Since
finite-rank perturbations do not change the ac and the essential spectrum
and $\sigma_{\rm ess}(H_2) \subset (-\infty, a] \cup [b, \infty)$, we can
assume w.l.o.g. that $\sigma(H_2) \subset  (-\infty, a +\epsilon/2]
\cup [b-\epsilon/2, \infty)$. \\

\vspace{0.5cm}

\noindent
\underline{Step 1:}  Define $\widehat{H}_1 := H_1 \chi_{[a+\epsilon,b-
\epsilon]}(H_1)$.  Then $[a + \epsilon, b -  \epsilon ] \subseteq
\sigma_{\rm ac}(\widehat{H}_V)$, where $ \widehat{H}_V := \left[ \begin{array}{cc}
 \widehat{H}_1 &  V  \\ V^*  &  H_2 \end{array} \right]$.

\vspace{0.5cm}

This follows directly from  Theorem \ref{thm:albeverio} by noting that
$\sigma(\widehat{H}_1) = [a+\epsilon,b-\epsilon]\subset \rho(H_2)$.
\vspace{0.5cm}

\noindent
\underline{Step 2:}  Let $f \in C_0^\infty([a + \epsilon, b - \epsilon])$.
Then  $f(H_V) - f(\widehat{H}_V)$ is trace class.

\vspace{0.5cm}

To show Step 2 we use almost analytic functional calculus. Let $\widetilde{f}
\in C_0^\infty(\C)$ be an almost analytic extension of $f$, satisfying
$\widetilde{f} |_{\R} = f$, $\widetilde{f}(x+iy) = 0$ if $x \notin  \supp f$,
and that for some  constant $C$ we have $\left|\partial_{\overline{z}}
\widetilde{f}(z) \right|  \leq C |{\rm Im} z |^3$ for all $z \in \mathbb{C}$.
[Of course, $\partial_{\overline{z}} := \frac12(\partial_{x}
+ i \partial_y)$ for $z=x+i y\in\C$ and $\overline{z}:=x-i y$.] Setting $R(A,z)
:= ( z - A)^{-1}$ for a bounded self-adjoint operator $A$ we recall the
Helffer--Sj\"ostrand formula (see \cite[(5)]{Da95})
\begin{equation}\label{HS}
f(A):=-  \frac{1}{\pi} \int_{\C}  \partial_{\overline{z}} \widetilde{f}(z)\,
R(A,z)\, dxdy\,.
\end{equation}
Setting $W := \left[ \begin{array}{cc} 0 & V \\ V^* & 0 \end{array}\right]$ and $\widehat{H}_0  := \left[ \begin{array}{cc} \widehat{H}_1 & 0 \\ 0  & H_2 \end{array}\right]$,
and applying the resolvent identity twice, we find
\begin{align*}
f(H_V) - f(\widehat{H}_V)
 &= - \frac{1}{\pi} \int_{\C} \partial_{\overline{z}} \widetilde{f} (z)
 \left[ R(H_0,z)  -  R(\widehat{H}_0,z) \right] \,  dx dy \\
 &  -\frac{1}{\pi} \int_{\C} \partial_{\overline{z}} \widetilde{f} (z)
 \left[ R(H_0,z)  W R(H_0,z)  -  R(\widehat{H}_0,z) W R(\widehat{H}_0,z)
 \right] \,  dx dy \\
 & -  \frac{1}{\pi}  \int_{\C} \partial_{\overline{z}} \widetilde{f} (z)
 \left[ R(H_0,z)  W R(H_0,z)  W R(H,z) - R(\widehat{H}_0,z)  W
 R(\widehat{H}_0,z)  W R(\widehat{H},z) \right] \,  dx dy \,.
\end{align*}
The first term on the right-hand side equals $f(H_0) - f(\widehat{H}_0)$ and
thus vanishes. The third term on the right-hand side is a trace class operator.
The second term also vanishes, as we now show. It is a combination of
terms of the following form,
\begin{align*}
- \frac{1}{\pi} \int_{\C} \partial_{\overline{z}} \widetilde{f} (z)  R(H_1,z)
V  R(H_2,z)\,dx dy
 &= - \frac{1}{\pi} \int_{\C} \partial_{\overline{z}} \widetilde{f} (z)
 R(H_1,z)  V \int_{\sigma(H_2)} \frac{1}{z - t } \, d P_{H_2}(t) \,dx dy \\
 &=  \int_{\sigma(H_2)}  f(H_1) \frac{1}{H_1-t}  V  \, d P_{H_2}(t) \,.
\end{align*}
In the first line we have applied the Spectral Theorem for $R(H_2,z)$ with
the spectral projections $P_{H_2}$ of $H_2$. In the last equality we have used that $\tilde{f}(z)(z-t)^{-1}$ is an almost analytic extension of $f(x)
(x - t)^{-1}$ since $\partial_{\overline{z}} (z-t)^{-1} = 0$ for  $z \neq t$,
and the Helffer--Sj\"ostrand formula (\ref{HS}).
By inspection,  the right-hand side of the last displayed formula does not
change if we replace $H_1$ by $\widehat{H}_1$.

\vspace{0.5cm}

\noindent
\underline{Step 3:}  $[a+\epsilon,b-\epsilon ] \in \sigma_{\rm ac}(H_V)$.
\vspace{0.5cm}

By the statement of Step 2 and the Kato-Rosenblum Theorem
\cite[Theorem XI.8]{RS3} we know that
$\sigma_{\rm ac}(f(H_V)) = \sigma_{\rm ac}(f(\widehat{H}_V))$ for all
$f \in C_0^\infty([a+\epsilon,b-\epsilon])$. By the Spectral Theorem we
conclude that $\sigma_{\rm ac}(H_V) \cap [a+\epsilon,b-\epsilon] =
\sigma_{\rm ac}(\widehat{H}_V) \cap [a+\epsilon,b-\epsilon] = [a+\epsilon,b-
\epsilon]$, where the second equality follows from the result of Step 1.\qed


\begin{thebibliography}{30}

\bibitem{AK08} S.~Albeverio and A.~Konstantinov:
{\em On the absolutely continuous spectrum of block operator matrices},
Math. Nachr. {\bf 281} (2008), no. 8, 1079--1087.

\bibitem{AMM03} S.~Albeverio, K.~Makarov, and A.~Motovilov:
{\em Graph subspaces and the spectral shift function},
Canad. J. Math. {\bf 55} (2003), no. 3, 449--503.

\bibitem{B02} J.~Bourgain:
{\em On random Schr\"odinger operators on $\Z^2$},
Discrete Contin. Dyn. Syst. {\bf 8}, no. 1 (2002), 1--15.

\bibitem{B03} J.~Bourgain:
{\em Random lattice Schr\"odinger operators with
decaying potential: some higher dimensional phenomena},
V.D.~Milman and G.~Schechtman (Eds.) LNM {\bf 1807}, 70--98, 2003.

\bibitem{CFKS}  H.~Cycon, R.~Froese, W.~Kirsch, B.~Simon,
{\em Schr\"odinger operators with application to quantum mechanics and global geometry}, Springer-Verlag, Berlin, 1987.


\bibitem{Da95} E.B.~Davies:
{\em The functional calculus},
J. London Math. Soc. (2) {\bf 52} (1995) 166--176.

\bibitem{DK99} P.~Deift and R.~Killip:
{\em On the absolutely continuous spectrum of one-dimensional Schr\"odinger
operators with square summable potentials},
Commun. Math. Phys. {\bf 203} (1999), no. 2, 341--347.

\bibitem{DSS85}
F.~Delyon, B.~Simon, and B.~Souillard:
{\em From power pure point to continuous spectrum in disordered systems},
Ann. Inst. H. Poincar\'e Phys. Th\'eor. {\bf 42} (1985), no. 3, 283--309.

\bibitem{D04} S.A.~Denisov:
{\em Absolutely continuous spectrum of multidimensional Schr\"odinger operator},
Int. Math. Res. Notices {\bf 2004} (2004), no. 74, 3963--3982.

\bibitem{D06} S.~Denisov:
{\em On the preservation of absolutely continuous spectrum for Schr\"odinger
operators}, J. Funct. Anal. {\bf 231} (2006), no. 1, 143--156.

\bibitem{D09} S.~Denisov: {\em On a conjecture by Y.~Last}, arXiv:0908.3681.

\bibitem{FHS07} R.~Froese, D.~Hasler, and W.~Spitzer:
{\em Transfer matrices, hyperbolic geometry and absolutely continuous
spectrum for some discrete Schr\"odinger operators on graphs},
J. Funct. Anal. {\bf 230} (2006), no. 1, 184--221.

\bibitem{FHS09} R.~Froese, D.~Hasler, and W.~Spitzer:
{\em Absolutely continuous spectrum for a random potential on a tree with
strong transverse correlations and large weighted loops},
Rev. Math. Phys. {\bf 21} (2009), 1--25.

\bibitem{KKO} W.~Kirsch, M.~Krisha, J.~Obermeit, {\em Anderson Model with decaying randomness-mobility edge},
Math. Z. {\bf 235} (2000), no. 3, 421-433

\bibitem{K98} A.~Klein:
{\em Extended states in the Anderson model on the Bethe lattice},
Adv. Math. {\bf 133} (1998), no. 1, 163--184.

\bibitem{KS} S.~Kotani, B.~Simon, {\em Stochastic Schr\"odinger operators and Jacobi matrices on the strip},
 Comm. Math. Phys. {\bf 119} (1988), no. 3, 403--429.

\bibitem{Kr} M.~Krishna, {\em Anderson model with decaying randomness-extended states}, Proc. Ind. Acad. Sci. (Math Sci) {\bf 100} (1990), no. 4, 285--294.

\bibitem{KrS} M.~Krishna, K.B.~Sinha, {\em Spectral properties of Anderson Type operators with decaying randomness}, Proc. Ind. Acad. Sci. (Math Sci) {\bf 111} (2001), no. 2, 179--201.

\bibitem{LNS} A.~Laptev, S.~Naboko, and O.~Safronov:
{\em A Szeg\H{o} condition for a multidimensional Schr\"odinger operator},
J. Funct. Anal. {\bf 219} (2005), no. 2, 285--305.

\bibitem{MV04} S.~Molchanov and B.~Vainberg:
{\em Schr\"odinger operators with matrix potentials. Transition from the
absolutely continuous to the singular spectrum},
J. Funct. Anal. {\bf 215} (2004), no. 1, 111--129.

\bibitem{RS3} M.~Reed and B.~Simon:
{\em Methods of Modern Mathematical Physics III: Scattering Theory},
Academic Press 1979.

\bibitem{Schulz-Baldes} H. Schulz-Baldes:
{\em Perturbation theory for Lyapunov exponents of an Anderson model on a
strip},
Geom. Funct. Anal. {\bf 14} (2004), 1089--1117.

\bibitem{Simon} B.~Simon:
{\em Schr\"odinger operators in the twenty-first century},
Mathematical Physics 2000 (eds. A. Fokas, A. Grigoryan, T. Kibble and
B. Zegarlinski), Imperial College Press, London, 283--288.

\bibitem{Stoll}
P.~Stollmann, {\em Caught by disorder. Bound states in random media. Progress
in Mathematical Physics},  Birkh\"auser, Boston,  2001.


\end{thebibliography}
\end{document}